\documentclass[11pt]{amsart}
\usepackage[a4paper,left=25mm,top=20mm,right=20mm,bottom=15mm]{geometry}

\usepackage{latexsym}
\usepackage{amsmath}
\usepackage{amssymb}
\usepackage{MnSymbol}
\usepackage{graphicx}

\usepackage[colorlinks]{hyperref}
\hypersetup{
	citecolor=blue,
   linkcolor=blue,
}

\DeclareMathOperator{\lca}{lca}

\providecommand{\keywords}[1]{\textbf{\textit{Keywords: }} #1}

\newtheorem{thm}{Theorem}
\newtheorem{lem}{Lemma}
\newtheorem{defn}{Definition}

\title[Undirected Fitch Graphs]{A Short Note on Undirected Fitch Graphs}

\author[ Gei{\ss}, Hellmuth, Long, and Stadler]{Manuela Gei{\ss} \and Marc Hellmuth \and Yangjing Long  \and Peter F.\ Stadler}

\address{Bioinformatics Group, Department of Computer Science, 
				 Universit{\"a}t Leipzig, H{\"a}rtelstrasse 16-18, D-04107 Leipzig,
				 Germany}
\email{manuela@bioinf.uni-leipzig.de}

\address{Department of Mathematics and
  Computer Science, University of Greifswald, Walther-Rathenau-Stra{\ss}e
  47, D-17487 Greifswald, Germany;
Center for Bioinformatics, Saarland University, Building E 2.1, P.O.\ Box
  151150, D-66041 Saarbr{\"u}cken, Germany}
\email{mhellmuth@mailbox.org}

\address{Department of Mathematics and
  Computer Science, University of Greifswald, Walther-Rathenau-Stra{\ss}e
  47, D-17487 Greifswald, Germany}
\email{yjlong@sjtu.edu.cn}

\address{Bioinformatics Group, Department of Computer Science, 
 Universit{\"a}t Leipzig, H{\"a}rtelstrasse 16-18, D-04107 Leipzig,
 Germany;
Interdisciplinary Center for
  Bioinformatics, the German Centre for Integrative Biodiversity Research 
  (iDiv) Halle-Jena-Leipzig the Competence Center for Scalable Data
  Services and Solutions Dresden-Leipzig the Leipzig Research Center for 
  Civilization Diseases, and the Centre for Biotechnology and Biomedicine at 
  Leipzig University; the Max Planck Institute for Mathematics in
  the Sciences, Leipzig, Germany; the Institute for Theoretical Chemistry, 
  University of Vienna, Vienna, Austria; the Center of noncoding RNA in
  Health and Technology (RTH) at the University of
  Copenhagen; and the Santa Fe Institute, Santa Fe, NM}
\email{studla@bioinf.uni-leipzig.de}

\thanks{This work is supported in part by the BMBF-funded project ``Center for
				RNA-Bioinformatics'' (031A538A, de.NBI/RBC)}



\date{\today}

\begin{document}

\maketitle
\begin{abstract}
  The symmetric version of Fitch's xenology relation coincides with class
  of complete multipartite graph and thus cannot convey any non-trivial
  phylogenetic information.
\end{abstract}




\sloppy

\bigskip \bigskip

\emph{Fitch graphs} \cite{Geiss:17a} form a class of directed graphs that
is derived from rooted, $\{0,1\}$-edge-labeled trees $T$ in the following
manner: The vertices of the Fitch graph are the leaves of $T$.  Two
distinct leaves $x$ and $y$ of $T$ are connected by an arc $(x,y)$ from $x$
to $y$ if and only if there is at least one edge with label $1$ on the
(unique) path in $T$ that connects the least common ancestor $\lca(x,y)$ of
$x$ and $y$ with $y$. Fitch graphs model ``xenology'', an important
  binary relation among genes introduced by Walter M.\ Fitch
  \cite{Fitch:00}. Interpreting $T$ as a phylogenetic tree and $1$-edges
as horizontal gene transfer events, the arc $(x,y)$ in the Fitch graph
encodes the fact that $y$ is xenologous with respect to $x$. The directed
Fitch graphs are the topic of Ref.\ \cite{Geiss:17a}, which among other
results proves a characterization in terms of eight forbidden induced
subgraphs.

It is natural to consider the symmetrized version of this relationship,
i.e., to interpret $\{x,y\}$ as a xenologous pair whenever the evolutionary
history separated $x$ and $y$ by at least one horizontal transfer event. In
mathematical terms, this idea is captured by
\begin{defn}
  \label{def:uFg}
  Let $T$ be a rooted tree with leaf set $X$ and let
  $\lambda:E(T)\to\{0,1\}$. Then the \emph{undirected Fitch graph} $F$
  explained by $(T,\lambda)$ has vertex set $X$ and edges $\{x,y\}\in E(F)$
  if and only if the (unique) path from $x$ to $y$ in $T$ contains at least
  one edge $e$ with $\lambda(e)=1$. A graph $F$ is an undirected Fitch
  graph if and only if it is explained in this manner by some edge-labeled
  rooted tree $(T,\lambda)$.
\end{defn}
 
Undirected Fitch graphs are closely related to their directed
counterparts. Since the path $\wp$ connecting two leaves $x$ and $y$ is
unique and contains their last common ancestor $\lca(x,y)$, there is a
1-edge along $\wp$ if and only if there is 1-edge on the path between $x$
and $\lca(x,y)$ or between $\lca(x,y)$ and $y$. The undirected Fitch graph
is therefore the underlying undirected graph of the directed Fitch graph,
i.e., it is obtained from the directed version by ignoring the direction of
the arcs.

The undirected Fitch graphs form a heritable family, i.e., if $F$ is a an
undirected Fitch graph, so are all its induced subgraphs. This is an
immediate consequence of the fact that directed Fitch graphs are a heritable
family of digraphs \cite{Geiss:17a}. The fact can also be obtained directly
by considering the restriction of $T$ to a subset of leaves. This obviously
does not affect the paths or their labeling between the remaining vertices.

Clearly $F$ does not depend on which of the non-leaf vertices in $T$ is the
root. Furthermore, a vertex $v$ with only two neighbors and its two
incident edges $e'$ and $e''$ can be replaced by a single edge $e$. The new
edge is labeled $\lambda(e)=0$ if both $\lambda(e')=\lambda(e'')=0$, and
$\lambda(e)=1$ otherwise. These operations do not affect the undirected
Fitch graph. Hence, we can replace the rooted tree $T$ by an unrooted tree
in Def.~\ref{def:uFg} and assume that all non-leaf edges have at least
degree $3$. To avoid trivial cases we assume throughout that $T$ has at
least two leaves and hence a Fitch graph has at least two vertices.

\begin{lem} 
  \label{lem:size3} 
  If $G$ is an undirected Fitch graph, then $F$ does not contain
  $K_1\cupdot K_2$ as an induced subgraph. In particular every undirected
  Fitch graph is a complete multipartite graph.
\end{lem}
\begin{proof}
  There is a single unrooted tree with three leaves, namely the star $S_3$,
  which admits four non-isomorphic $\{0,1\}$-edge labelings defined by the
  number $N$ of 1-edges. The undirected Fitch graphs $F_N$ are easily
  obtained. In the absence of 1-edges, $F_0=\overline{K_{3}}$ is edge-less.
  For $N=2$ and $N=3$ there is 1-edge along the path between any two
  leaves, i.e., $F_2=F_3=K_3$. For $N=1$ one leaf is connected to the other
  two by a path in $S_3$ with an 1-edge; the path between the latter two
  leaves consists of two 0-edges, hence $F_1=P_3$, the path of length two.
  Hence, only three of the four possible undirected graphs on three
  vertices can be realized, while $K_1\cupdot K_2$ is not an undirected
  Fitch graph.  By heredity, $K_1\cupdot K_2$ is therefore a forbidden
  induced subgraph for the class of undirected Fitch graphs. Finally, it is
  well known that the class of graphs that do not contain $K_1\cupdot K_2$
  as an induced subgraph are exactly the complete multipartite graphs, see
  e.g.\ \cite{Zverovich:99}.
\end{proof}

In order to show that forbidding $K_1\cupdot K_2$ is also sufficient, we
explicitly construct the edge-labeled trees necessary to explain complete
multipartite graphs. We start by recalling that each complete multipartite
graph $K_{n_1,\dots, n_k}$ is determined by its independent sets
$V_1,\dots,V_k$ with $|V_i|=n_i$ for $1\le i\le k$. By definition,
$\{x,y\}\in E(K_{n_1,\dots, n_k})$ if and only if $x\in V_i$ and $y\in V_j$
with $i\ne j$. In particular, therefore, $K_{n_1,\dots, n_k}$ with at least
two vertices is connected if and only if $k\ge 2$.  The complete
$1$-partite graphs are the edge-less graphs $\overline K_{n}$.

Since $K_1\cupdot K_2$ is an induced subgraph of the path on four vertices
$P_4$, any graph $G$ that does not contain $K_1\cupdot K_2$ as an induced
subgraph must be $P_4$-free, i.e., a cograph \cite{Corneil:81}. Cographs
are associated with vertex-labeled trees known as cotrees, which in turn
are a special case of modular decomposition trees \cite{Gallai:67}. The
cotrees of connected multipartite graphs have a particularly simple shape,
illustrated without the vertex labels in Fig.~\ref{fig:cmgtree}.  The
cotree has a root labeled ``1'' and all inner vertices labeled ``0''. Here
we do not need the connection between cographs and their cotrees,
however. Therefore we introduce these trees together with an edge-labeling
that is useful for our purposes in the following
\begin{defn} 
  For $k=1$, $T[n]$ is the star graph $S_n$ with $n$ leaves.  
  For $k\ge 2$, the tree $T[n_1,\dots,n_k]$ has a root $r$ with $k$
  children $c_i$, $1\le i\le k$. The vertex $c_i$ is a leaf if
  $|V_i|=n_i=1$ and has exactly $n_i$ children that are leaves if
  $|V_i|=n_i\ge 2$. 
  \newline 
  For $k=1$ all edges $e$ of $T[n]$ are labeled $\lambda^*(e)=0$.  For
  $k\ge 2$ we set $\lambda^*(\{r,c_i\})=1$ for $1\le i\le k$ and
  $\lambda^*(e)=0$ for all edges not incident to the root.
  \label{def:FT}
\end{defn}

\begin{figure}
  \begin{center}
    \includegraphics[width=0.7\textwidth]{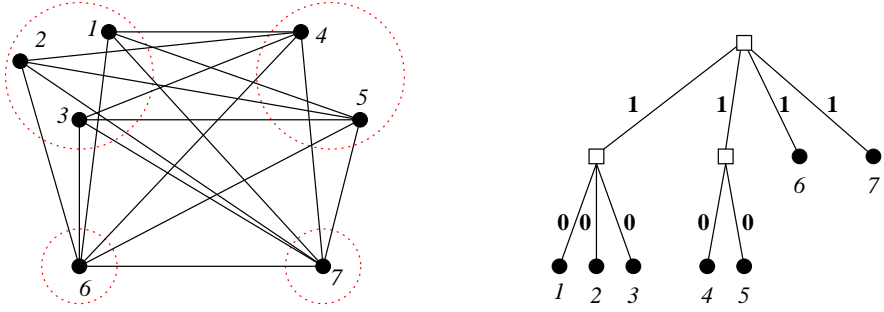}
  \end{center}
  \caption{The complete multipartite graph $K_{3,2,1,1}$ is the Fitch graph
    explained by the tree $T[3,2,1,1]$ with edge labeling $\lambda^*$ shown
    with bold numbers \textbf{0} and \textbf{1}.}
  \label{fig:cmgtree}
\end{figure}

Now we can prove our main result:
\begin{thm}
  \label{thm:char}
  A graph $G$ is an undirected Fitch graph if and only if it is a complete
  multipartite graph. In particular, $K_{n_1,\dots,n_k}$ is explained by 
  $(T[n_1,\dots,n_k],\lambda^*)$. 
\end{thm} 
\begin{proof}
  Lemma \ref{lem:size3} implies that an undirected Fitch graph is a
  complete multipartite graph. To show the converse, we fix an arbitrary 
  complete multipartite graph $G = K_{n_1,\dots,n_k}$ and find an 
  edge-labeled rooted tree $(T, \lambda^*)$ that explains $G$. 

  For $k=1$ it is trivial that $(T[n],\lambda^*)$ explains
  $\overline{K_n}$.

  For $k\ge 2$ consider the tree $T[n_1,\dots n_k]$ with edge labeling
  $\lambda^*$ and let $F$ be the corresponding Fitch graph.  The leaf-set
  of $T[n_1,\dots n_k]$ partitioned into exactly $k$ subsets $L_1$, \dots,
  $L_k$ defined by (a) singletons adjacent to the root and (b) subsets
  comprising at least two leaves adjacent to the same child $c_i$ of the
  root. Furthermore, we can order the leaf sets so that $|L_i|=n_i$.  By
  construction, all vertices within a leaf set $L_i$ are connected by a
  path that does not run through the root and thus, contains only 0-edges,
  if $|L_i|>1$ and no edge, otherwise.  The $L_i$ are independent sets in
  $F$. On the other hand any two leaves $x\in L_i$ and $y\in L_j$ with
  $i\ne j$ are connected only by path through the root, which contains two
  1-edges. Thus $x$ and $y$ are connected by an edge in $F$. Thus $F$ is a
  complete multipartite graph of the form
  $K_{|L_1|,\dots,|L_k|}=K_{n_1,\dots,n_k}$. Since $K_{n_1,\dots,n_k}$ is
  explained by $(T[n_1,\dots,n_k],\lambda^*)$ for all $n_i\ge 1$ and all
  $k\ge 2$, and $\overline{K_n}$ is explained by $(T[n],\lambda^*)$, we
  conclude that every complete multipartite graph is a Fitch graph.
\end{proof}

Complete multipartite graphs $G=(V,E)$ obviously can be recognized in
$O(|V|^2)$ time e.g., by checking that its complement is a disjoint union
of complete graphs), and even in $O(|V|+|E|)$ time by explicitly
constructing its modular decomposition tree \cite{McConnell:99}. Given the
tree $T[n_1,\dots,n_k]$, the canonical edge labeling $\lambda^*$ is then
assigned in $O(|V|)$ time.

A tree $(T,\lambda)$ that explains a Fitch graph $F$ is \emph{minimum} if
it has the smallest number of vertices among all trees that explain $F$.
In this case, $(T,\lambda)$ is also \emph{least-resolved}, i.e., the
contraction of any edge in $(T,\lambda)$ results in a tree that does not
explain $F$.  Not surprisingly, the tree $T[n_1,\dots,n_k]$ is almost
minimum in most, and minimum in some cases: Since the vertices of the Fitch
graph must correspond to leaves of the tree, $T[n_1,\dots,n_k]$ is
necessarily minimum whenever it is a star, i.e., for $T[n]$ and
$T[1,\dots,1]$. In all other cases, its only potentially ``superfluous''
part is its root.  Indeed, exactly one of the edges connecting the root
with a non-leaf neighbor can be contracted without changing the
corresponding Fitch graph. It is clear that this graph is minimal: The leaf
sets $L_i$ must be leaves of an induced subtree without an intervening
1-edge. Having all vertices of $L_i$ adjacent to the same vertex is
obviously the minimal choice. Since the $L_i$ must be separated from all
other leaves by a 1-edge, at least one neighbor of $c_i$ must be a
1-edge. Removing all leaves incident to a 0-edge results in a tree with at
least $k$ vertices that must contain at least $k-1$ 1-edges, since every
path between leaves in this tree must contain a 1-edge. The contraction of
exactly one of the $k$ 1-edges incident to the root $r$ in
$T[n_1,\dots,n_k]$ indeed already yields a minimal tree. In general, the
minimal trees are not unique, see Fig.\ \ref{fig:nonU-l}.

\begin{figure}
  \begin{center}
    \includegraphics[width=1.0\textwidth]{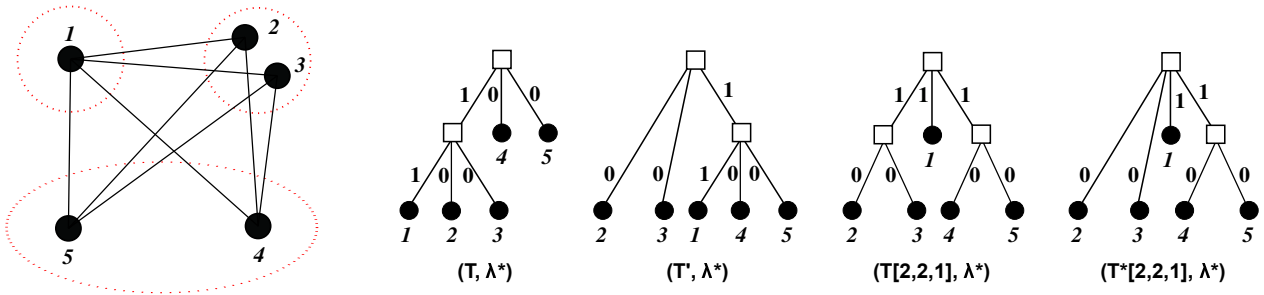}
  \end{center}
  \caption{The non-isomorphic trees $(T,\lambda^*)$, $(T',\lambda^*)$
    $(T[2,2,1],\lambda^*)$, and $(T^*[2,2,1],\lambda^*)$ all explain the
    same complete multipartite graph $K_{2,2,1}$. Three of these trees have
    the smallest possible number (7) of vertices and thus are
    minimal. These can be obtained from the tree
      $(T[2,2,1],\lambda^*)$ specified in Definition~\ref{def:uFg} by
      contracting of one of its inner 1-edge and possibly re-routing the
      resulting tree.
  }
  \label{fig:nonU-l}
\end{figure}

The practical implication of Thm.~\ref{thm:char} in the context of
phylogenetic combinatorics is that the mutual xenology relation cannot
convey any interesting phylogenetic information: Since the undirected Fitch
graphs are exactly the complete multipartite graphs, which in turn are
completely defined by their independent sets, the only insight we can gain
by considering mutual xenology is the identification of the maximal subsets
of taxa that have not experienced any horizontal transfer events among
them.

\bibliographystyle{plain}
\bibliography{symfitch}

\end{document}